\documentclass[11pt,a4paper,twoside,english,11pt]{elsarticle}
\usepackage[T1]{fontenc}
\usepackage[latin9]{inputenc}
\usepackage{amsthm}
\usepackage{amsmath}
\usepackage{amssymb}
\usepackage{graphicx}
\usepackage{esint}

\makeatletter

\theoremstyle{plain}
\newtheorem{thm}{\protect\theoremname}
  \theoremstyle{definition}
  \newtheorem{defn}[thm]{\protect\definitionname}
  \theoremstyle{plain}
  \newtheorem{cor}[thm]{\protect\corollaryname}

\usepackage{psfrag}\usepackage{babel}

\usepackage[displaymath,mathlines]{lineno}

\@ifundefined{showcaptionsetup}{}{%
 \PassOptionsToPackage{caption=false}{subfig}}
\usepackage{subfig}

\usepackage{babel}\providecommand{\propositionname}{Proposition}
  \providecommand{\remarkname}{Remark}
\providecommand{\theoremname}{Theorem}

\theoremstyle{plain}
  \theoremstyle{plain}
  \newtheorem{prop}[thm]{\protect\propositionname}
  \theoremstyle{remark}

\@ifundefined{showcaptionsetup}{}{%
 \PassOptionsToPackage{caption=false}{subfig}}
\usepackage{subfig}
\makeatother

\usepackage{babel}
  \providecommand{\corollaryname}{Corollary}
  \providecommand{\definitionname}{Definition}
\providecommand{\theoremname}{Theorem}

\begin{document}

\title{A simple probabilistic construction yielding generalized entropies
and divergences, escort distributions and $q$-Gaussians\footnote{This is a preprint version that differs from the published version, Physica A doi:10.1016/j.physa.2012.04.024, in minor revisions, pagination and typographics details.}}

\author{J.-F. Bercher}

\ead{jf.bercher@esiee.fr}

\address{Université Paris-Est, LIGM, UMR CNRS 8049, ESIEE-Paris\\
 5 bd Descartes, 77454 Marne la Vallée Cedex 2, France\\
 tel: 33-1-45-92-65-15 fax: 33-1-45-92-66-99}
\begin{abstract}
We give a simple probabilistic description of a transition between
two states which leads to a generalized escort distribution. When
the parameter of the distribution varies, it defines a parametric
curve that we call an escort-path. The Rényi divergence appears as
a natural by-product of the setting. We study the dynamics of the
Fisher information on this path, and show in particular that the thermodynamic
divergence is proportional to Jeffreys' divergence. Next, we consider
the problem of inferring a distribution on the escort-path, subject
to generalized moments constraints. We show that our setting naturally
induces a rationale for the minimization of the Rényi information
divergence. Then, we derive the optimum distribution as a generalized
$q$-Gaussian distribution. 
\end{abstract}
\maketitle
\date{Typesetted \today}


\section{Introduction}

In this paper, we 
give a simple probabilistic description of a transition between two
states, which leads to a parametric curve in the form of a generalized
escort distribution. We call escort-path this parametric curve. In
this setting, we show that the Rényi information divergence emerges
naturally as a characterization of the transition. Along this escort-path,
we study the Fisher information. In particular, we show that the thermodynamic
divergence on the escort-path is proportional to Jeffreys' divergence.
Finally, we consider the inference of a distribution subject to moments
computed with respect to the escort distribution. First, we show that
our setting leads to a rationale for the minimization of the Rényi
information divergence. Then, we derive the optimum distribution as
a generalized Gaussian distribution. 

Before going into the details of the results, we shall present the
context and introduce the main definitions on our main ingredients,
that is the escort distributions, information divergences, and Fisher
information. 

Throughout the paper, we will work with univariate probability densities
defined with respect to a general measure $\mu(x)$ on a set $X$.
For instance, the Shannon-Boltzmann entropy will be expressed as 
\begin{equation}
H[f]=-\int f(x)\log f(x)\mathrm{d}\mu(x).\label{eq:ShannonEntropy}
\end{equation}
As particular cases, we have that if $X$ is the real line and $\mu$
the Lebesgue measure, then the expression above corresponds to the
differential entropy. When the set $X$ is $\mathbb{N}$ or a subset
of $\mathbb{N}$ and $\mu$ the counting measure, then the expression
reduces to the standard discrete entropy. When $\mu$ is a probability
measure, then the expression (\ref{eq:ShannonEntropy}) can also be
seen as the relative entropy from the measure with density $f$\,
to the measure $\mu$. 

Let us now turn to the notion of escort distribution. If $f(x)$ is
an univariate probability density with respect to $\mu(x)$, then
we define its escort distribution of order $q$, $q\geq0,$ by
\begin{equation}
f_{q}(x)=\frac{f(x)^{q}}{\int f(x)^{q}\mathrm{d}\mu(x)},\label{eq:escort_f-1}
\end{equation}
provided that $M_{q}[f]=\int f(x)^{q}\mathrm{d}\mu(x)$ is finite.
 These escort distributions have been introduced as an operational
tool in the context of multifractals \cite{chhabra_direct_1989},
\cite{beck_thermodynamics_1993}, with interesting connections with
the standard thermodynamics. Discussion of their geometric properties
can be found in \cite{abe_geometry_2003,ohara_dually_2010}. Escort
distributions also prove to be useful in source coding where they
enable to derive optimum codewords with a length bounded by the Rényi
entropy \cite{bercher_source_2009}. 

The results presented in this paper are connected to the nonextensive
statistical physics introduced by Tsallis, see e.g. \cite{tsallis_introduction_2009}.
Indeed, the nonextensive statistical physics uses a generalized entropy,
makes use of escort distributions and exhibit generalized Gaussians.
All these elements will pop up in our construction, which, therefore
could lead to new viewpoints or interpretations in this context. It
is particularly remarkable that the derivation of the maximum Tsallis
entropy distributions in nonextensive thermostatistics requires a
constraint in the form of an ``escort mean value'', that is computed
with respect to an escort distribution like (\ref{eq:escort_f-1})
\cite{abe_necessity_2005,tsallis_escort_2009}.  

One can immediately extend the notion of escort distribution to deal
with two probability densities $f(x)$ and $g(x)$ as follows. 
\begin{defn}
\label{def:generalized_escort}Let $f$ and $g$ be two densities
with respect to a common measure $\mu,$ with $g$ dominated by $f$.
For $q\geq0$ such that $M_{q}[f,g]=\int f(x)^{q}g(x)^{1-q}d\mu(x)<\infty$,
we call generalized escort distribution the function 
\begin{equation}
f_{q}(x)=\frac{f(x)^{q}g(x)^{1-q}}{\int f(x)^{q}g(x)^{1-q}\mathrm{d}\mu(x)}.\label{eq:generalized_escort}
\end{equation}
We will also denote, when non ambigous, by $E_{q}[.]$ the statistical
expectation with respect to the generalized escort distribution with
index $q.$ 

This generalized escort distribution is simply a weighted geometric
mean of $f(x)$ and $g(x),$ and reduces to $f_{q}(x)=f(x)$ for $q=1$
and to $f_{q}(x)=g(x)$ for $q=0.$ Obviously, if $g(x)$ is a uniform
density whose support includes the support of $f(x),$ then the generalized
escort distribution gives back the standard one (\ref{eq:escort_f-1}).
Actually, the generalized escort (\ref{eq:generalized_escort}) appeared
in Chernoff analysis of the efficiency of hypothesis tests \cite{chernoff_measure_1952},
and enables to define the best achievable exponent in the bayesian
probability of error \cite[Chapter 11]{cover_elements_2006}. As
$q$ varies, the generalized escort distribution defines a curve that
connects $f(x)$ to $g(x)$ and further. In the general framework
of information geometry \cite{amari_methods_2000}, the generalized
escort distribution (\ref{eq:generalized_escort}) coincides with
the geodesic joining $f$ and $g$ in the case of an exponential connection.
Such interpretation also appeared in a work by Campbell \cite{campbell_relation_1985}.

Throughout the paper, we will focus on the generalized escort distribution
and the path it defines, that we will call the escort-path.  
\end{defn}

Distances between probability distributions will be measured by means
of information divergences. We will use the Kullback-Leibler directed
information divergence which is defined as follows. 
\begin{defn}
Let $f$ and $g$ be two univariate densities with respect to a common
measure $\mu$, with $f$ absolutely continuous with respect to $g.$
The Kullback-Leibler directed information divergence is given by 
\begin{equation}
D(f||g)=\int f(x)\log\frac{f(x)}{g(x)}\mathrm{d}\mu(x).\label{eq:KLdiv}
\end{equation}
 
\end{defn}
It is understood, as usual, that $0\log0=0\log0/a=0\log0/0=0$. Note
that if we take $g(x)=1$ in the expression above, then we obtain
minus the Shannon entropy $H[f]$. Let us also recall that the minimization
of the Kullback-Leibler divergence is a well established inference
method, analog to Jaynes' maximum entropy approach and which is supported
in particular by large deviation results \cite{ellis_theory_1999}.
We will also make use of the Rényi information divergence introduced
in \cite{rnyi_measures_1961}. 
\begin{defn}
Let $f$ and $g$ be two probability densities with respect to a measure
$\mu.$ If $f$ is absolutely continuous with respect to $g$, then,
for $q\geq0$ such that $M_{q}[f,g]=\int f(x)^{q}g(x)^{1-q}d\mu(x)<\infty$,
the Rényi divergence is defined by 
\begin{equation}
D_{q}(f||g)=\frac{1}{q-1}\log\int f(x)^{q}g(x)^{1-q}\mathrm{d\mu(}x).\label{eq:RenyiDiv}
\end{equation}

\end{defn}
Let us recall that the divergence is always non negative $D_{q}(f||g)\geq0$
with the equality sign iff $f=g.$ By L'Hôpital's rule, the Kullback
divergence is recovered in the limit $q\rightarrow1$. Taking $g(x)=1$
in the expression of the Rényi divergence yields the negative of the
Rényi entropy, noted $H_{q}[f].$ 

We will study Fisher information along the escort-path. Indeed, it
is well known that the Fisher information metric is a Riemannian metric
that can be defined on a smooth statistical manifold \cite{rao_information_1945,burbea_differential_1985}.
Furthermore, the Fisher information serves as a measure of the information
about a parameter in a distribution. It has intricate relationships
with maximum likelihood and has many implications in estimation theory,
as exemplified by the Cramér-Rao bound which provides a fundamental
lower bound on the variance of an estimator \cite{barankin_locally_1949}.
It is also used as a method of inference and understanding in statistical
physics and biology, as promoted by Frieden \cite{frieden_physicsfisher_2000,frieden_sciencefisher_2004}. 
\begin{defn}
Let $f(x;\theta)$ denote a probability density with respect to a
measure $\mu$, where $\theta$ is a real parameter, and suppose that
$f(x;\theta)$ is differentiable with respect to $\theta$. Then,
the Fisher information in the density $f$ about the parameter $\theta$
is defined as
\begin{equation}
I[f,\theta]=\int_{\mathcal{}}\left(\frac{\partial\ln f(x;\theta)}{\partial\theta}\right)^{2}f(x;\theta)\mathrm{d}\mu(x).\label{eq:GenFisher}
\end{equation}

\end{defn}

The remaining of the paper is structured as follows. In section \ref{sec:The-escort-path}
we show that the generalized escort presented above arises naturally
in a simple probabilistic description of a transition between two
states. Interestingly, the Rényi information divergence, and in a
particular case the Rényi entropy, emerges as a characterization of
the transition.

In section \ref{sec:Fisher-information-along}, we study the Fisher
information, with respect to $q$, along the escort-path.  
 We show in particular that the integral of the Fisher information
along the path, the thermodynamic divergence, is proportional to Jeffreys'
divergence.

In section \ref{sec:Paths-with-minimum}, we consider the problem
of inferring the distribution $f(x)$ in (\ref{eq:escort_f-1}) or
(\ref{eq:generalized_escort}) on the escort-path when the only available
information is given as a mean value. This mean value is the statistical
expectation taken with respect to an escort distribution: this is
the escort mean value used in nonextensive statistics. Different possible
approaches, such as minimizing the directed divergence, or Jeffreys
divergence or the thermodynamic divergence, reduce to the minimization
of the Rényi information divergence. In this case, the probability
distribution that emerges is a generalized Gaussian distribution,
which is particularly important in applications. 

\section{\label{sec:The-escort-path}The escort-path}

It has been observed that Tsallis' extended thermodynamics seems particularly
appropriate in the case of deviations from the classical Boltzmann-Gibbs
equilibrium. This suggests that the original MaxEnt formulation ``find
the closest distribution to a reference under a mean constraint\textquotedblright{}\ may
be amended by introducing a new constraint that displaces the equilibrium.
The partial or displaced equilibrium can be imagined as an equilibrium
characterized by two distributions, say $p_{0}(x)$ and $p_{1}(x)$.
Instead of selecting the nearest distribution to a reference under
a mean constraint, we may look for a distribution $p_{q}(x)$ simultaneously
close, in some sense, to two distinct references: such a distribution
will be localized somewhere `between' $p_{0}(x)$ and $p_{1}(x)$.

\subsection{Displaced equilibrium }

\begin{figure}
\begin{centering}
\subfloat[Case $\eta<D(p_{1}||p_{0})$]{\begin{centering}
\psfrag{eta}{$\eta$}  
\psfrag{D(p||p1)}{\small $D(p||p_1)$}   
\psfrag{D(p||p0)}{\small $D(p||p_0)$}  
\psfrag{p0}{\small $p_0$}   
\psfrag{p1}{\small $p_1$}   
\psfrag{pt}{\small $p_q$}  

\psfrag{p}{\small $p$}  \includegraphics[width=6cm]{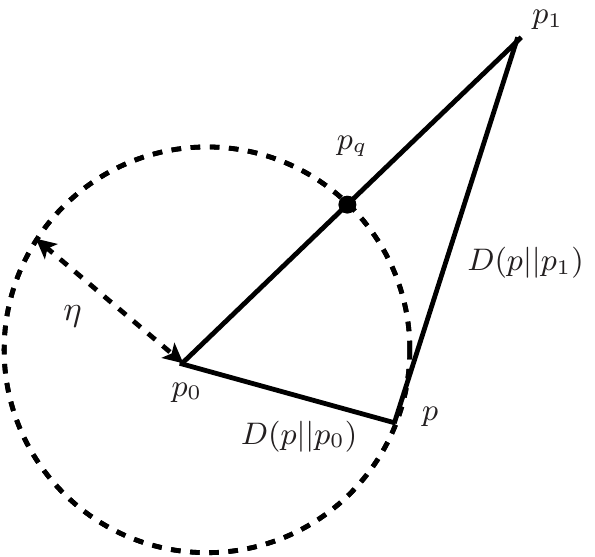}
\par\end{centering}

}~~~~~~ \subfloat[\label{fig:Caseb}Case $\eta>D(p_{1}||p_{0})$]{\centering{}\psfrag{eta}{$\eta$}  
\psfrag{D(p||p1)}{\small $D(p||p_1)$}   
\psfrag{D(p||p0)}{\small $D(p||p_0)$}  
\psfrag{p0}{\small $p_0$}   
\psfrag{p1}{\small $p_1$}   
\psfrag{pt}{\small $p_q$}  
\psfrag{p}{\small $p$}  \includegraphics[width=5cm]{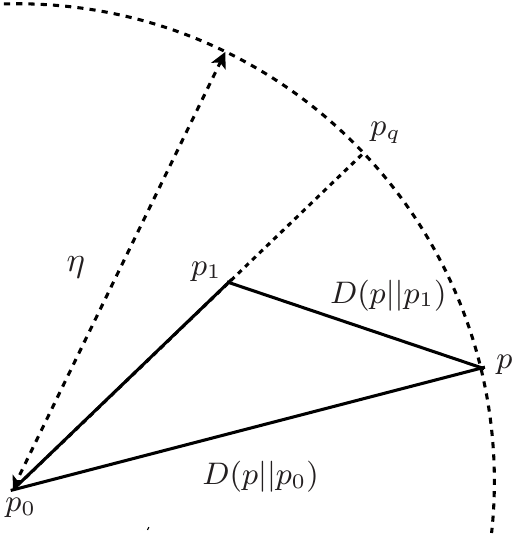}}
\par\end{centering}

\caption{\label{fig:Constrained-equilibrium}Constrained equilibrium between
states $p_{0}$ and $p_{1}$: the equilibrium distribution is sought
in the set of all distributions such that $D(p|[p_{0})=\eta,$ and
with minimum Kullback distance to $p_{1}$. The equilibrium distribution
$p_{q}$ , the generalized escort distribution, is ``aligned'' with
$p_{0}$ and $p_{1}$ and intersects the set $D(p|[p_{0})=\eta.$ }
\end{figure}

We consider two equilibrium states with respective probability densities
$p_{0}(x)$ and $p_{1}(x)$ with respect to a common measure $\mu,$
at some point $x$ in the phase space, and we look at intermediate
states defined by the following scenario. The system with initial
state $p_{0}$, subject to a generalized force, is moved at a distance
$\eta=D(p||p_{0})$ from $p_{0},$ where $D(p||p_{0})$ is the Kullback-Leibler
divergence (or relative entropy) from $p$ to $p_{0}.$ Then, the
system is attracted toward the final state $p_{1}.$ Therefore, the
new intermediate equilibrium state, say $p_{q},$ is chosen as the
one which minimizes its divergence to the attractor $p_{1}$ while
being hold on at the distance $\eta$ from $p_{0}.$ As illustrated
in Figure\,\ref{fig:Constrained-equilibrium}, the intermediate probability
density is located on the ``straight line'' $p_{0}-p_{1}$ and intersects
the circle with radius $\eta$ centered at $p_{0}.$ More precisely,
the problem can be written as follows: 
\begin{equation}
\left\{ \begin{array}{c}
\min_{p}~D(p||p_{1})\\
\text{s.t. }D(p||p_{0})=\eta\\
\text{{and}}\int p(x)\mathrm{d}\mu(x)=1
\end{array}\right.\label{eq:minD}
\end{equation}
 where ``s.t.'' stands for ``subject to'', and where the Kullback-Leibler
divergence $D(f||g)$ is defined by (\ref{eq:KLdiv}). The solution
is given by the following Theorem. 
\begin{thm}
\label{pro:Prop_EscorGen} 

Let $p_{1}$ a probability density function with respect to $\mu,$
and $p_{0}$ a non negative function. Assume that $p_{1}$ is absolutely
continuous with respect to $p_{0}.$ Let  $p_{q}$ denote the generalized
escort distribution with index $q\geq0$ 
\begin{equation}
p_{q}(x)=\frac{p_{1}(x)^{q}p_{0}(x)^{1-q}}{\int p_{1}(x)^{q}p_{0}(x)^{1-q}\mathrm{d}\mu(x)},\label{P-star}
\end{equation}
with $M_{q}(p_{1},p_{0})=\int p_{1}(x)^{q}p_{0}(x)^{1-q}\mathrm{d}\mu(x)<\infty$.
If $E_{q}\left[\log\frac{p_{1}}{p_{0}}\right]$ is finite, where $E_{q}\left[.\right]$
denote the statistical expectation with respect to $p_{q},$ and if
$q$ is chosen such that $D(p_{q}||p_{0})=\eta,$ then the generalized
escort distribution (\ref{P-star}) is the unique solution of problem
(\ref{eq:minD}). 
\end{thm}
\begin{proof} Let us evaluate the divergence $D(p||p_{q}).$ For
all densities $p$ satisfying $D(p||p_{0})=\eta$, we have \begin{small}
\begin{align}
\hspace{-0.5cm}D(p||p_{q}) & = & \int\!\! p(x)\log\frac{p(x)}{p_{q}(x)}\mathrm{d}\mu(x)=\int\!\! p(x)\log\frac{p(x)^{q}p(x)^{1-q}}{p_{1}(x)^{q}p_{0}(x)^{1-q}}\mathrm{d}\mu(x)+\log M_{q}(p_{1},p_{0})\\
 & = & q\,\int\!\! p(x)\log\frac{p(x)}{p_{1}(x)}\mathrm{d}\mu(x)+(1-q)\int\!\! p(x)\log\frac{p(x)}{p_{0}(x)}\mathrm{d}\mu(x)+\log M_{q}(p_{1},p_{0})\label{eq:trois}\\
 & = & q\, D(p||p_{1})+(1-q)\eta+\log M_{q}(p_{1},p_{0})\label{eq:trois-1}
\end{align}
 \end{small} Observe that $D(p_{q}||p_{0})=qE_{q}\left[\log\frac{p_{1}}{p_{0}}\right]-\log M_{q}$
and that $D(p_{q}||p_{1})=(1-q)E_{q}\left[\log\frac{p_{1}}{p_{0}}\right]-\log M_{q}$
so that both divergences exist. Therefore, taking $p=p_{q},$ the
last equality gives 
\begin{equation}
D(p_{q}||p_{q})=q\, D(p_{q}||p_{1})+(1-q)\eta+\log M_{q}(p_{1},p_{0}).\label{eq:quatre}
\end{equation}
 Finally, subtracting (\ref{eq:trois-1}) and (\ref{eq:quatre}) yields
\[
D(p||p_{q})-D(p_{q}||p_{q})=q\,\left(D(p||p_{1})-D(p_{q}||p_{1})\right).
\]
 Since $q\geq0$ and since $D(p||p_{q})\geq0$ with equality iff $p=p_{q},$
we obtain that $D(p||p_{1})\geq D(p_{q}||p_{1})$ which proves the
Theorem. \end{proof} It is interesting to note that (\ref{P-star})
is nothing else but a generalized version of the \emph{escort} or
\emph{zooming} distribution of nonextensive thermostatistics, and
that the corresponding statistical expectations are the so-called
escort-means or generalized averages. Obviously, one recovers a standard
escort distribution like (\ref{def:generalized_escort}) when $p_{0}(x)$
is uniform with respect to $\mu$. This is immediate if $\mu$ has
a compact support. However, if one wants to use a uniform measure
on the whole real axis, with $\mu$ the Lebesgue measure, then such
a measure is no more a probability density since it integrates to
infinity. In such case, it still possible to modify the formulation
to include this case as well. Indeed, with $p_{0}(x)=1$, the expression
of the Kullback-Leibler divergence $D(p||p_{0})$ becomes nothing
but minus the standard entropy
\[
H[p]=-\int p(x)\log p(x)\mathrm{d}\mu(x).
\]
Therefore, the problem turns into the research of a distribution with
a given entropy, which minimizes the divergence to $p_{1}$:

\begin{equation}
\left\{ \begin{array}{c}
\min_{p}~D(p||p_{1})\\
\text{s.t. }H[p]=-\eta\\
\text{{and}}\int p(x)\mathrm{d}\mu(x)=1.
\end{array}\right.\label{eq:minD-1}
\end{equation}
This setting can be illustrated as was done in Figure\,\ref{fig:Constrained-equilibrium},
excepted that the circle now corresponds to the set of distributions
with a given level of entropy. Observe that neither the Theorem \ref{pro:Prop_EscorGen}
nor its proof require that $p_{0}$ is a probability density. Therefore
we can take $p_{0}(x)=1$ and obtain the solution of (\ref{eq:minD-1})
as a simple corollary. 
\begin{cor}
\label{prop:EscortClassica}Let $p_{q}$ denote the escort distribution
with index $q$, associated with $p_{1},$ defined by
\begin{equation}
p_{q}(x)=\frac{p_{1}(x)^{q}}{\int p_{1}(x)^{q}\mathrm{d\mu(}x)},\label{P-star-1}
\end{equation}
provided that $M_{q}(p_{1})=\int p_{1}(x)^{q}\mathrm{d}\mu(x)<\infty$.
If $E_{q}\left[\log p_{1}\right]$ is finite, where $E_{q}\left[.\right]$
denote the statistical expectation with respect to $p_{q},$ and if
$q$ is chosen such that $H[p_{q}]=-\eta,$ then the escort distribution
(\ref{P-star-1}) is the unique solution of problem (\ref{eq:minD-1})\textup{.} 
\end{cor}
When $q$ varies, the function $\eta(q)=D(p_{q}||p_{0})$ is monotonically
increasing, and particular intermediate values satisfy the implicit
relationship $D(p_{q}||p_{0})=\eta$. This property will be proved
in section \ref{sec:Fisher-information-along}, corollary \ref{cor:monotoneta},
 as a simple consequence of a result on Fisher information. For $q=0$
we have $\eta=0$ and for $q=1,$ we have $\eta=D(p_{1}||p_{0})$.
Accordingly, as $q$ varies, $p_{q}$ traces out a curve, the \emph{escort-path,}
that connects $p_{0}$ ($q=0)$ and $p_{1}$ ($q=1$). In the case
$q>1,$ we have $\eta>D(p_{1}||p_{0})$ as shown in Figure\,\ref{fig:Caseb}.



Interestingly enough, recent results have shown that the average dissipated
work during a transition can be expressed as a relative entropy \cite{kawai_dissipation_2007,parrondo_entropy_2009}.
Along these lines, with an Hamiltonian even in the momenta, the minimization
of $D(p||p_{1})$ may be understood as a minimization of the average
dissipated work for a transition from $p$ to $p_{1}.$

\subsection{Rényi and Jeffreys' divergences as by-products}

It is interesting to outline that the Rényi divergence and entropy
arise as a by-product of our construction. Indeed, the minimum of
the Kullback-Leibler divergence can be expressed as follows. 
\begin{cor}
The minimum divergence is given by 
\begin{equation}
D(p_{q}||p_{1})=\left(1-\frac{1}{q}\right)\left(\eta-D_{q}(p_{1}||p_{0})\right)\label{eq:optvalue}
\end{equation}
 where $D_{q}(p_{1}||p_{0})$ is the {Rényi information divergence}
with index $q,$ from $p_{1}$ to $p_{0}$.  
\end{cor}
\begin{proof} By direct calculation from the expression of the solution
$p_{q}(x)$, or by a direct consequence of relation (\ref{eq:quatre}).
\end{proof} 

If $p_{0}$ is a uniform distribution, then $-D_{q}(p_{1}||p_{0})=H_{q}(p_{1}),$
the Rényi entropy, $p_{q}$ is the standard escort distribution and
(\ref{eq:optvalue}) becomes 
\[
D(p_{q}||p_{1})=\left(1-\frac{1}{q}\right)\left(\eta+H_{q}(p_{1})\right).
\]

Although it is convenient to think of the Kullback-Leibler divergence
$D(f||g)$ (\ref{eq:KLdiv}) as a distance between $f$ and $g$,
it is not symmetric and does not satisfy the triangle inequality.
Kullback and Leibler themselves introduced a symmetrized version,
which was also considered before by Jeffreys. This Jeffreys' divergence
appears here to be a simple affine function of Rényi information divergence
$D_{q}(p_{1}||p_{0}).$ 
\begin{cor}
\label{prop:JeffreysRenyi}The Jeffreys divergence between $p_{1}$
and the generalized escort distribution $p_{q}$ is given by 
\begin{equation}
J(p_{1},p_{q})=D(p_{1}||p_{q})+D(p_{q}||p_{1})=\frac{(q-1)^{2}}{q}\left(D_{q}(p_{1}||p_{0})-\eta\right).\label{eq:JeffreysRenyi}
\end{equation}

\end{cor}
\begin{proof} This is a simple consequence of (\ref{eq:trois-1}),
which gives $D(p_{1}||p_{q})=(1-q)\eta+\log M_{q}(p_{1},p_{0})$ if
$p=p_{1},$ and of (\ref{eq:quatre}) that gives $D(p_{q}||p_{1})=(1-\frac{1}{q})\eta-\frac{1}{q}\log M_{q}(p_{1},p_{0})$.
\end{proof} 

As an interesting consequence, we see that if one wants to minimize
the symmetric divergence between $p_{1}$ and $p_{q},$ subject to
additional constraints, then this simply amounts to the minimization
of the Rényi information divergence with the same constraints. When
$p_{0}$ is uniform, this becomes the maximization of the Rényi entropy,
or equivalently of the Tsallis entropy. It is thus interesting that
our setting induces both an escort distribution and a Rényi divergence
(or entropy), and besides with a common index $q$. Actually, although
these two quantities are essential ingredients in nonextensive statistical
mechanics, their relationships are discussed, e.g. \cite{pennini_semiclassical_2007}.

\section{\label{sec:Fisher-information-along}Fisher information along the
escort-path}

Suppose now that $p_{0}(x)$ and $p_{1}(x)$ depend on a parameter
$\theta.$ The Fisher information metric is based on the Fisher information
matrix on a vector parameter $\theta$ attached to a density $p(x;\theta)$.
This Fisher information matrix has entries

\[
\left[I(\theta)\right]_{i,j}=\int p(x;\theta)\left(\frac{\partial}{\partial\theta_{i}}\log p(x;\theta)\right)\left(\frac{\partial}{\partial\theta_{j}}\log p(x;\theta)\right)\mathrm{d\mu(}x).
\]
 The derivative of the logarithm of the density with respect to the
parameter is called the score function. The mean of the score function
is zero, so that the Fisher information matrix is the covariance of
the score function.

The length of a curve parametrized by $t,$ from $0$ to $T,$ is
given by 
\[
\mathcal{{L}}=\sum_{i}\sum_{j}\int_{0}^{T}\sqrt{\frac{\mathrm{d}\theta_{i}}{\mathrm{d}t}\left[I(\theta)\right]_{i,j}\frac{\mathrm{d}\theta_{j}}{\mathrm{d}t}}\,\mathrm{d}t.
\]
 In the context of thermodynamics, this quantity is called the thermodynamic
length \cite{weinhold_metric_1975,crooks_measuring_2007,shenfeld_minimizing_2009}.
A related quantity is the thermodynamic divergence, or energy of the
curve, given by 
\[
\mathcal{{J}}=\sum_{i}\sum_{j}\int_{0}^{T}\frac{\mathrm{d}\theta_{i}}{\mathrm{d}t}\left[I(\theta)\right]_{i,j}\frac{\mathrm{d}\theta_{j}}{\mathrm{d}t}\,\mathrm{d}t.
\]
 By Jensen's inequality, we have immediately that $\mathcal{J}\geq\mathcal{L}^{2}$.
An interesting point, that outlines the importance of these quantities,
is the fact that the thermodynamic divergence asymptotically bounds
the dissipation induced by a finite time transformation of a thermodynamic
system \cite{nulton_quasistatic_1985,crooks_measuring_2007}. Hence,
it is interesting here to study some characteristics of the Fisher
information along the escort-path. The general study of the Fisher
information on the escort-path with respect to a general parameter
$\theta$ is interesting in its own right. However, in order to save
space, we will focus here on a special case. Let us still simply mention
that when $p_{0}$ is uniform, the related Fisher information is the
\emph{escort-Fisher information} which has been considered in \cite{hammad_mesure_1978,pennini_rnyi_1998,pennini_escort_2004}.

As we have seen, the generalized escort distribution describes a geometric
path, the escort-path, connecting distributions $p_{0}$ and $p_{1}$
for the values $q=0$ and $q=1$.  Clearly, the densities on the
escort-path are characterized by the index $q$. Hence it is quite
natural to evaluate the distance between two densities on the path,
as well as the Fisher information with respect to $q.$ Let us begin
by a general expression of the Fisher information on the path. Then,
we will be able to link this Fisher information to information divergences
on the path. 

\begin{thm}
Let $p_{q}$ be the generalized escort distribution as in (\ref{P-star}).
Then, the Fisher information with respect to $q$ of the generalized
escort distribution is given by
\begin{equation}
I(q)=\int\frac{1}{p_{q}(x)}\left(\frac{dp_{q}(x)}{dq}\right)^{2}\mathrm{d\mu(}x)=\int\frac{dp_{q}(x)}{dq}\log\frac{p_{1}(x)}{p_{0}(x)}\,\mathrm{d\mu(}x)\label{eq:Jqdef1}
\end{equation}
provided that $E_{r}\left[\left(\log\frac{p_{1}}{p_{0}}\right)^{2}\right]$
is finite for $r$ in a compact neighborhood of $q$. The Fisher information
with respect to $q$ can also be written as the variance of the log-likelihood
ratio: 
\begin{equation}
I(q)=E_{q}\left[\left(\log\frac{p_{1}(x)}{p_{0}(x)}-E_{q}\left[\log\frac{p_{1}(x)}{p_{0}(x)}\right]\right)^{2}\right].\label{eq:Jqscore}
\end{equation}
 \end{thm}
\begin{proof}
The second order moment condition on the log-likelihood ratio implies,
by Jensen inequality, that both $E_{q}\left[\left|\log\frac{p_{1}}{p_{0}}\right|\right]$
and $E_{q}\left[\log\frac{p_{1}}{p_{0}}\right]$ are finite. Let us
first consider $M_{q}(p_{1},p_{0})=\int p_{1}(x)^{q}p_{0}(x)^{1-q}\mathrm{d}\mu(x).$
The integrand is clearly differentiable with respect to $q,$ and
this derivative, which is equal to $p_{q}\log\frac{p_{1}}{p_{0}}$
is continuous and is absolutely integrable since $E_{q}\left[\left|\log\frac{p_{1}}{p_{0}}\right|\right]$
is finite. Furthermore, by the second order moment hypothesis, the
last expression is also locally integrable with respect to $q.$ This
enables to use Leibniz' rule and differentiate under the integral
sign, which gives 
\begin{equation}
\frac{\mathrm{d}\log M_{q}}{\mathrm{d}q}=\int p_{q}(x)\log\frac{p_{1}(x)}{p_{0}(x)}\mathrm{d\mu(}x)=E_{q}\left[\log\frac{p_{1}(x)}{p_{0}(x)}\right].\label{eq:dlogM}
\end{equation}
Then, by direct calculation, we also have 
\begin{equation}
\frac{\mathrm{d}p_{q}(x)}{\mathrm{d}q}=p_{q}(x)\left(\log\frac{p_{1}(x)}{p_{0}(x)}-E_{q}\left[\log\frac{p_{1}(x)}{p_{0}(x)}\right]\right),\label{eq:dpqdq}
\end{equation}
which, inserted in the definition of the Fisher information in (\ref{eq:Jqdef1})
gives (\ref{eq:Jqscore}). 

By (\ref{eq:dpqdq}), we have that 
\begin{equation}
\int\left|\frac{\mathrm{d}p_{q}(x)}{\mathrm{d}q}\right|\mathrm{d}\mu(x)\leq E_{q}\left[\left|\log\frac{p_{1}}{p_{0}}\right|\right]+\left|E_{q}\left[\log\frac{p_{1}}{p_{0}}\right]\right|<\infty.\label{eq:eqq}
\end{equation}
Moreover, by the second order moment hypothesis, (\ref{eq:eqq}) is
also locally integrable with respect to $q.$ Since $\int p_{q}(x)\mathrm{d}\mu(x)=1,$
then by Leibniz' rule we get that $\frac{\mathrm{d}}{\mathrm{d}q}\int p_{q}(x)\mathrm{d}\mu(x)=\int\frac{\mathrm{d}p_{q}(x)}{\mathrm{d}q}\mathrm{d}\mu(x)=0.$
Finally, the right hand side of (\ref{eq:Jqdef1}) is obtained by
using (\ref{eq:dpqdq}) and the fact that 
\[
\int\frac{\mathrm{d}p_{q}(x)}{\mathrm{d}q}\frac{\mathrm{d}\log M_{q}}{\mathrm{d}q}\mathrm{d}\mu(x)=\frac{\mathrm{d}\log M_{q}}{\mathrm{d}q}\int\frac{\mathrm{d}p(x)}{\mathrm{d}q}\mathrm{d}\mu(x)=0.
\]

\end{proof}
As a simple consequence, we can now check that $\eta=D(p_{q}||p_{0})$
is indeed a monotone increasing function of $q,$ as announced in
section \ref{sec:The-escort-path}. 
\begin{cor}
\label{cor:monotoneta} Let $p_{q}$ be a generalized escort distribution,
with $q>0,$ and assume that $E_{r}\left[\left(\log\frac{p_{1}}{p_{0}}\right)^{2}\right]<\infty$
for $r$ in a compact neighborhood of $q$. Then $\eta(q)=D(p_{q}||p_{0})$
is a strictly monotone increasing function of $q,$ with 
\begin{equation}
\frac{\partial}{\partial q}\eta(q)=q\, I(q)>0\label{eq:derivKullback}
\end{equation}
\end{cor}
\begin{proof}
Note that $\eta(q)=\int p_{q}(x)\log\frac{p_{q}(x)}{p_{0}(x)}\text{d}\mu(x)=q\int p_{q}(x)\log\frac{p_{1}(x)}{p_{0}(x)}\text{d}\mu(x)-\log M_{q}.$
Under the second order moment condition, one can differentiate under
the integral sign, take into account (\ref{eq:dlogM}) and it remains
\[
\frac{\partial}{\partial q}\eta(q)=q\frac{\partial}{\partial q}\int p_{q}(x)\log\frac{p_{1}(x)}{p_{0}(x)}\text{d}\mu(x)=q\int\frac{dp_{q}(x)}{dq}\log\frac{p_{1}(x)}{p_{0}(x)}\text{d}\mu(x),
\]
where we recognize the Fisher information in (\ref{eq:Jqdef1}). Therefore,
taking into account the fact that both $q$ and the Fisher information
are positive, we get (\ref{eq:derivKullback}).
\end{proof}
Finally, an important result is that the integral of the Fisher information,
the ``energy'' of the curve, is nothing but the Jeffreys divergence.
This result is mentioned in \cite{dabak_relations_????}. Alternatively,
this can also be obtained as a consequence of the general integral
representation of the Kullback-Leibler divergence \cite[eq. 3.71]{amari_methods_2000}.
We propose here a direct proof of the result.
\begin{thm}
\label{prop:ThermoDiv}Let $p_{r}$ and $p_{s}$ be two generalized
escort distributions. Assume that $E_{q}\left[\left(\log\frac{p_{1}}{p_{0}}\right)^{2}\right]<\infty$
for all $q\in[r,s].$ Then, the integral of the Fisher information
along the escort-path, from $q=r$ to $q=s$ is proportional to Jeffreys'
divergence between $p_{r}$ and $p_{s}:$ 
\begin{equation}
\left(s-r\right)\int_{r}^{s}I(q)\mathrm{d}q=J(p_{s},p_{r})=D(p_{s}||p_{r})+D(p_{r}||p_{s}).\label{eq:JeffreyFisher1}
\end{equation}
 With $r=0$ and $s=1$, we get the integral along the whole path
connecting $p_{0}$ and $p_{1},$ that is
\begin{equation}
\int_{0}^{1}I(q)\mathrm{d}q=J(p_{1},p_{0})=D(p_{1}||p_{0})+D(p_{0}||p_{1}).\label{eq:JeffreyFisher2}
\end{equation}

\end{thm}
\begin{proof} The Fisher information is finite on the escort path;
therefore its integral over a compact interval is also finite. Let
us integrate the right equality in (\ref{eq:Jqdef1}): 
\begin{equation}
\int_{r}^{s}I(q)\mathrm{d}q=\int_{r}^{s}\int\frac{\mathrm{d}p_{q}(x)}{\mathrm{d}q}\log\frac{p_{1}(x)}{p_{0}(x)}\,\mathrm{d}\mu(x)\,\mathrm{d}q.\label{eq:fubini}
\end{equation}
 Since $I(q)$ is positive and $\int_{r}^{s}I(q)\mathrm{d}q$ finite,
it is possible to apply Fubini's theorem to the right hand side of
(\ref{eq:fubini}), and exchange the order of integrations. Thus,
integrating with respect to $q$ yields 
\begin{equation}
\int_{r}^{s}I(q)\mathrm{d}q=\int\left(p_{s}(x)-p_{r}(x)\right)\log\frac{p_{1}(x)}{p_{0}(x)}\,\mathrm{d}\mu(x).\label{eq:integralJq}
\end{equation}
 On the other hand, the divergence $D(p_{s}||p_{r})$ writes
\[
D(p_{s}||p_{r})=(s-r)\int p_{s}(x)\log\frac{p_{1}(x)}{p_{0}(x)}\,\mathrm{d}\mu(x)-\log M_{s}+\log M_{r},
\]
 and similarly for $D(p_{r}||p_{s}).$ Adding the two divergences
and taking into account (\ref{eq:integralJq}) give the result (\ref{eq:JeffreyFisher1}).
\end{proof} 

Finally, let $\theta_{i},$ $i=1..M$ denote a set of intensive variables,
which are some functions of the index $q.$ Then, we have that $\frac{d\log p}{dq}=\sum_{i=1}^{M}\frac{\partial\log p}{\partial\theta_{i}}\frac{d\theta_{i}}{dq}$
and the Fisher information with respect to $q$ can be expressed as
\[
I(q)=\int p(x)\left(\frac{\mathrm{d}\log p(x)}{\mathrm{d}q}\right)^{2}\mathrm{d}\mu(x)=\sum_{i=1}^{M}\sum_{j=1}^{M}\frac{\mathrm{d}\theta_{i}}{\mathrm{d}q}\left[I(\theta)\right]_{i,j}\frac{\mathrm{d}\theta_{j}}{\mathrm{d}q},
\]
 where $I(\theta)$ is the Fisher information matrix with respect
to $\theta$. Therefore, for the escort-path we introduced, we obtain
that the thermodynamic divergence is nothing but the Jeffreys divergence:
\begin{equation}
\mathcal{J}=\int_{0}^{1}I(q)\mathrm{d}q=\sum_{i=1}^{M}\sum_{j=1}^{M}\int_{0}^{1}\frac{\mathrm{d}\theta_{i}}{\mathrm{d}q}\left[I(\theta)\right]_{i,j}\frac{\mathrm{d}\theta_j}{\mathrm{d}q}\,\mathrm{d}q=D(p_{1}||p_{0})+D(p_{0}||p_{1}).\label{eq:ThermoDivergenceTotale}
\end{equation}

\section{\textmd{\normalsize \label{sec:Paths-with-minimum}}{\normalsize Inference
of a distribution subject to $q$-moments constraints }}

In the last section of this paper, we investigate some relationships
between escort-distributions, information divergences, Fisher information and generalized Gaussians.
Let us return to the model of states transition as presented in section
\ref{sec:The-escort-path} that led us to the generalized escort distribution
(\ref{P-star}) as the optimum intermediate between $p_{0}$ and $p_{1}$. 

Assume that the distribution $p_{1}$ is not exactly known but that
the available information is given as an expectation under the escort
$p_{q}.$ This expectation is the so-called generalized expectation,
or $q$-average which is largely used in nonextensive statistics,
although it is generalized here with the presence of $p_{0}$. In
our context, it has the clear meaning of an expectation with respect
to the intermediate distribution $p_{q}$ at a given distance of a
reference $p_{0},$ c.f. Theorem \ref{pro:Prop_EscorGen}, or with
a given entropy, c.f. Corollary \ref{prop:EscortClassica}. Let the
observable be given as the absolute moment of order $\alpha$: 
\begin{equation}
m_{\alpha,q}[p_{1}]=E_{q}\left[|x|^{\alpha}\right]=\frac{\int|x|^{\alpha}p_{1}(x)^{q}p_{0}(x)^{1-q}\mathrm{d}\mu(x)}{\int p_{1}(x)^{q}p_{0}(x)^{1-q}\mathrm{d}\mu(x)}.\label{eq:qmoment}
\end{equation}
 Typically, the observable could be a mean energy, where the statistical
mean is taken with respect to the escort distribution. Then, the question
that arises is the determination of a general distribution $p_{1}$
compatible with this constraint.

One may keep the idea of minimizing the divergence to $p_{1},$ as
in the original problem (\ref{eq:minD}) which led us to the generalized
escort distribution. Since the Kullback divergence is a directed divergence,
we shall keep the notion of direction by minimizing $D(p_{q}||p_{1})$
for $q<1$ and $D(p_{1}||p_{q})$ for $q>1.$ In both cases, the divergence
is an affine function of the Rényi divergence $D_{q}(p_{1}||p_{0})$,
c.f. (\ref{eq:optvalue}). Therefore, these minimizations are finally
equivalent to the minimization of the Rényi divergence under the generalized
mean constraint. 

In the same vein, we may consider the minimization of the symmetric
Jeffreys' divergence between $p_{q}$ and $p_{1}.$ We have noticed
(\ref{eq:JeffreysRenyi}) that this divergence is also an affine function
of the Rényi divergence $D_{q}(p_{1}||p_{0})$. Therefore, its minimization
is also equivalent to the minimization of the Rényi divergence under
the generalized mean constraint. 

Finally, a natural idea is to select the distribution $p$, thus its
escort $p_{q},$ so as to minimize the thermodynamic divergence $\int_{q}^{1}I(t)\mathrm{d}t$
or $\int_{1}^{q}I(t)\mathrm{d}t$ from $p_{q}$ to $p$, while satisfying
the constraint (\ref{eq:qmoment}). We have seen that Jeffreys' divergence
$J(p_{1},p_{q})$ is proportional to the thermodynamic divergence,
as indicated in (\ref{eq:JeffreyFisher1}). As a consequence, the
minimization of the thermodynamic divergence between $p_{q}$ and
$p_{1}$ is also equivalent to the minimization of the Rényi information
divergence $D_{q}(p_{1}||p_{0}).$

It is known \cite{tsallis_introduction_2009} that the maximization
of Rényi entropy subject to generalized $q$-moments constraints,
or equivalently of Tsallis entropy under the same constraints, leads
to generalized Gaussian distributions. As far as the minimization
of the Rényi information divergence is concerned, a direct proof based
on a simple inequality can be derived along the lines in \cite[Appendix 1]{bercher_entropy_2008}
or in \cite[Proposition 4]{bercher_escort_2011}. Therefore, we have
the following result. 

\begin{prop} \label{prop:MinJeffreysDistribution}Among all distributions
with a given $q$-moment of order $\alpha$ as in (\ref{eq:qmoment}),
the distribution $p$ with minimum thermodynamic divergence, or equivalently
which minimizes Jeffreys' or Rényi divergence to its escort, is a
generalized Gaussian distribution given by
\begin{equation}
p(x)=\begin{cases}
\frac{1}{Z_{q}(\gamma)}\left(1-(1-q)\gamma|x|^{\alpha}\right)_{+}^{\frac{1}{1-q}}p_{0}(x) & \text{ for }q\neq1\\
\frac{1}{Z_{1}(\gamma)}\exp\left(-\gamma|x|^{\alpha}\right)p_{0}(x) & \mbox{ for }q=1,
\end{cases}\label{eq:qgauss2}
\end{equation}
 where we use the notation $\left(x\right)_{+}=\max\left\{ x,0\right\} ,$
and where $Z_{q}(\gamma)$ is the normalization factor. 

\end{prop} 

When $p_{0}$ is uniform the distribution becomes the standard Gaussian
distribution, for $\alpha=2,$ in the limit case $q=1$, by l'Hôpital's
rule. This gives the rationale for the denomination of ``generalized
Gaussians''. For $q<1,$ the probability density has a compact support,
while for $q>1$, the probability density has heavy tails with a power-law
behavior and is analog to a Student distribution. These generalized
Gaussians appear in statistical physics, where they are the maximum
entropy distributions of the nonextensive thermostatistics \cite{tsallis_introduction_2009}.
In this context, these distributions have been observed to present
a significant agreement with experimental data, and also to be the
analytical solution of actual physical problems \cite{lutz_anomalous_2003,schwaemmle_q-gaussians_2008},
\cite{ohara_information_2010}. In an other field, the generalized
Gaussians are the one dimensional instances of explicit extremal functions
of Sobolev, log-Sobolev or Gagliardo\textendash{}Nirenberg inequalities
on $\mathbb{R}^{n},$ with $n\geq2$ \cite{del_pino_best_2002}. 

Finally, let us close this paper with the example of a $q$-variance
constraint, i.e. $m_{2,q}[p]=\sigma_{q}^{2}$, with $p_{0}(x)=1$
and $\mu$ the Lebesgue measure. 
We have seen that among all distributions with a given differential
entropy, the standard escort distribution $p_{q}$ minimizes the Kullback-Leibler
divergence to $p,$ for some value of the index $q$ (Proposition
\ref{prop:EscortClassica}). If $p$ is free but its escort is known
to have a given variance, then the distribution $p$ which minimizes
the thermodynamic divergence or Jeffreys' divergence (Proposition
\ref{prop:MinJeffreysDistribution}), or equivalently that maximizes
the Rényi entropy, is the generalized Gaussian (\ref{eq:qgauss2})
with $\alpha=2$. In this setting, $p_{q}$ is located at the intersection
of the set of distributions with a given variance and of the set of
distributions with a given Shannon differential entropy. When $q$
varies, the optimum distributions follow a path indexed by $q$ which
is nothing but the path followed by the generalized Gaussians, with
compact support for $q<1$ and infinite support for $q>1.$ In the
limit case $q=1$, we obtain a standard Gaussian distribution, which
is its own escort distribution, and that has the maximum entropy among
all escort distributions with the same variance. These situations
are illustrated in Figure\,\ref{fig:FishPath}. 

\begin{figure}[h]
\begin{centering}
\includegraphics[width=9.2cm]{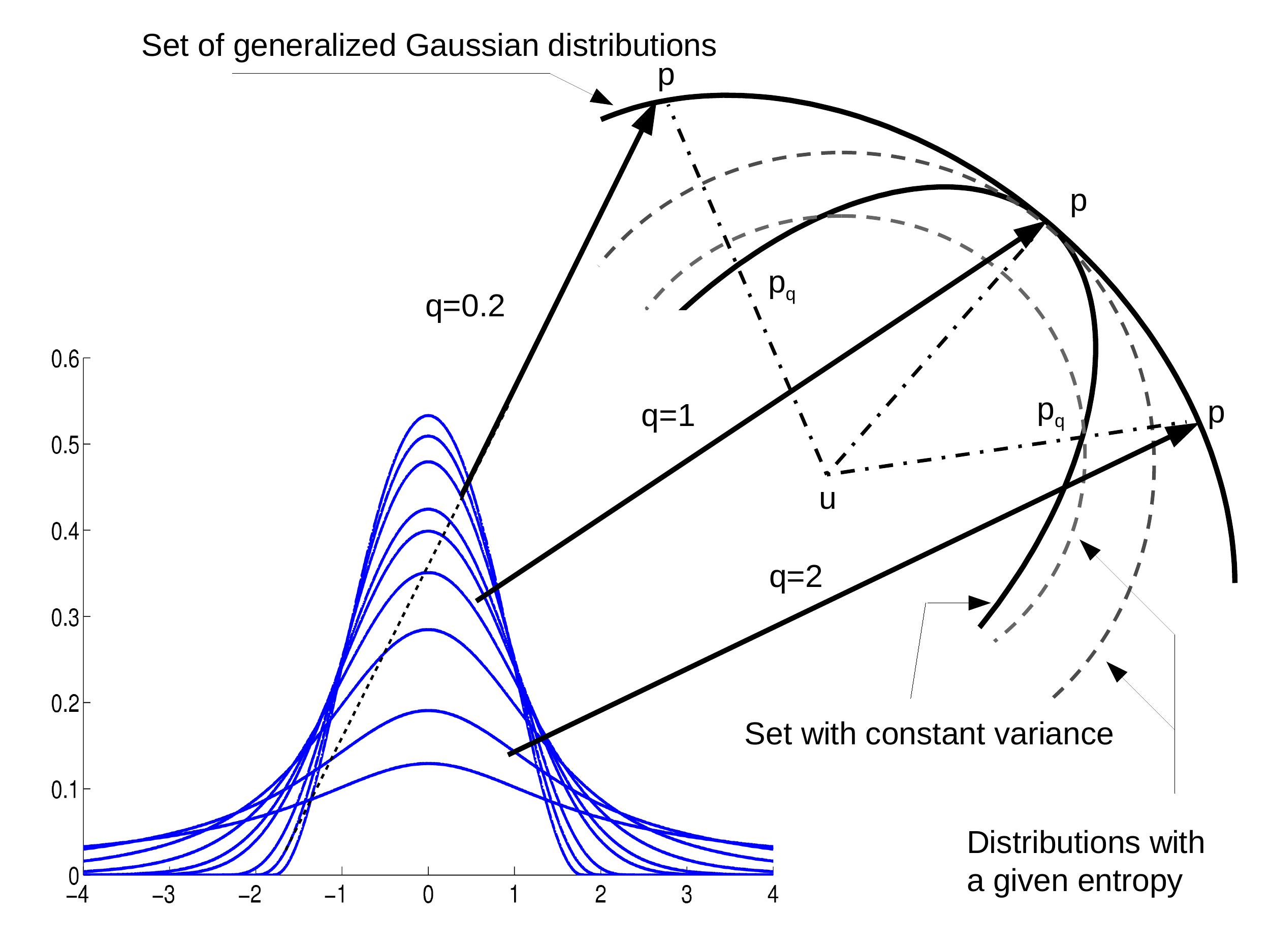} 
\par\end{centering}

\caption{\label{fig:FishPath} Path of distributions with maximum Rényi entropy
and fixed $q$-variance. For each value of $q,$ the optimum distribution
$p$ whose escort $p_{q}$ has a given variance is a generalized Gaussian.
Thus, when $q$ varies, the path followed by $p$ is the manifold
of generalized Gaussians with index $q$. }
\end{figure}

\section{Conclusions}

In this paper, we have presented a simple probabilistic model of transition
between two states, which leads naturally to a generalized escort
distribution. This generalized escort distribution enables to describe
a path, the escort-path, that connects the two states. Then, we have
connected several information measures, and studied their evolution
along the escort-path. In particular, we have obtained that the Rényi
information divergence appears naturally as a characterization of
the transition, and that the notion of escort mean values, as used
in nonextensive thermostatistics, receives a clear interpretation.
We have studied the properties and the evolution of Fisher information
along the escort-path. In particular, we have shown that the thermodynamic
divergence on the escort-path is a simple function of Jeffreys divergence.
We have also considered the problem of inferring a distribution on
the escort-path, subject to a moment constraint on its escort. Looking
for the distribution as the minimizer of the thermodynamic divergence,
we have shown that this procedure is equivalent to the minimization
of Rényi divergence subject to a $q$-moment constraint, which gives
a rationale for this approach. Finally, we have recalled that generalized
Gaussian distributions arise as solutions of the previous problem. 

Beyond the intrinsic interest of our geometric construction, which
enables to connect several quantities of information theory, we have
also pointed out possible connections with finite thermostatistics.
Furthermore, we have indicated that our findings interrelates several
ingredients of the nonextensive statistics. Let us also add that the
literature usually points out that the standard entropy (or divergence)
is a particular case of generalized Rényi or Tsallis entropies. Our
setting suggests a possible additional layer where the generalized
quantities are derived from a construction involving the classical
information measures. Therefore, we believe that the presented construction,
and the series of observations we made can be useful to workers in
this field. Future work should consider the extension of this setting
in the multivariate case. In future work, we plan to look for possible
connections with finite time thermodynamics. We also intend to study
the information theoretic relationships between generalized moments,
Fisher information and generalized Gaussians. 


\end{document}